\documentclass[11pt,twoside]{article}

\usepackage[ansinew]{inputenc} %%Siehe UTF-8
\usepackage{amsmath,amsfonts,amssymb,amsthm}
\usepackage{pstricks,pst-node,pst-coil,pst-plot,pstricks-add}
\usepackage{graphics,geometry,epsfig}
\usepackage{bbm}
\usepackage{floatflt}

\newcommand{\field}[1]{\mathbb{#1}}
\newcommand{\N}{\field{N}}
\newcommand{\R}{\field{R}}
\newcommand{\C}{\field{C}}

\newcommand{\HH}{\mathcal H}
\newcommand{\LL}{\mathcal L}
\newcommand{\EE}{\mathcal E}
\newcommand{\TT}{\mathcal T}

\newcommand{\eps}{\varepsilon}
\newcommand{\ph}{\varphi}

\newcommand{\ran}{\mathrm{Ran}}

\newcommand{\spa}{\mathrm{span}}

\newcommand{\restricted}{|\grave{}\,}
\newcommand{\disp}{\displaystyle}

\newcommand{\norm}[1]{\mbox{$\left\| #1 \right\|$}}           % new norm
\newcommand{\sprod}[2]{\mbox{$\left\langle #1,#2 \right\rangle$}}        % scalar product
\newcommand{\dist}{\operatorname{dist}}

\newcommand{\tr}{\operatorname{tr}}

\newtheorem{theorem}{Theorem}[section]
\newtheorem{lemma}[theorem]{Lemma}
\newtheorem{corollary}[theorem]{Corollary}

\newtheorem{proposition}[theorem]{Proposition}
\theoremstyle{plain}

\title{Unique Solutions to Hartree-Fock Equations\\ for Closed Shell Atoms}

\author{
Marcel Griesemer and Fabian Hantsch\\
Universit\"at Stuttgart, Fachbereich Mathematik\\
70550 Stuttgart, Germany}

\date{}
\begin{document}
\maketitle
\begin{abstract}
In this paper we study the problem of uniqueness of solutions to the 
Hartree and Hartree-Fock equations of atoms.
We show, for example, that the Hartree-Fock ground state of
a closed shell atom is unique provided
the atomic number $Z$ is sufficiently large compared to the number $N$
of electrons. More specifically,
a two-electron atom with atomic number $Z\geq 35$ has a unique
Hartree-Fock ground state given by
two orbitals with opposite spins and identical spatial wave functions.
This statement is wrong for some $Z>1$, which exhibits a phase segregation. 
\end{abstract}

\section{Introduction}

The Hartree-Fock method is widely used in
quantum chemistry for approximate electronic structure
computations \cite{Helgaker}. In the simplest case the state of the electrons is
described by a single Slater determinant and one seeks to minimize the
energy by variation of the one-electron orbitals. 
This is done by some self-consistent
field algorithm such as the Roothaan, or the level-shift
algorithm \cite{CB2000}. While there has been remarkable progress recently in the
analysis of convergence properties of these algorithms \cite{CB2000, Cances2000}, it is
still poorly understood what is actually being approximated. To
some extent this is due to our ignorance about the set of critical
points of the Hartree-Fock functional: we know that the Hartree-Fock functional
for a neutral atom has a minimizer and infinitely many other critical points
\cite{LS1977, Lions1987}, but uniqueness of the
minimizer, e.g., is not known even in cases where it is expected.
Neither is it known whether distinct methods for
finding critical points lead to distinct critical points. Our goal
is to give answers to such questions.

In this paper we establish existence and uniqueness of solutions
to the Hartree-Fock equations for positively charged atoms with
prescribed filled shells. We consider both restricted and
unrestricted Hartree-Fock theory. In the unrestricted case a
filled shell refers to a set of $qn^2$ electrons which means that
we take the number of all electrons, $N$, of the form $N=q\sum_{k=1}^{s}n_k^2$ where $1\leq
n_1<n_2\ldots <n_{s}$ and where $q$ denotes the number of spin states.
In the restricted case our notion of a shell is the usual one and
hence all values of $N$ that occur in noble gas atoms are
admissible. For atoms with partially filled shells our uniqueness
results will not hold. Our method is based on a perturbation
argument which exploits the fact that in the limit of large atomic numbers, $Z\to\infty$,
the electron-electron interaction energy is negligible compared to
the total energy. This forces us to choose $Z$ much larger than
$N$, but thanks to a novel technique for comparing spectral
projections, values of $Z$ as small as 35 can be handled in the case of
two-electron atoms. Also, we provide an example which shows
that $Z>N-1$ is not sufficient for our results to hold in general.
As a byproduct of our methods we obtain uniqueness of
the minimizer of the $N$-electron Hartree functional for
sufficiently large $Z$. 

The Hartree-Fock equations for an atom with atomic number $Z>0$ and $N$ electrons are the set of equations 
\begin{equation}\label{HFeq}
 \left(-\Delta -\frac{2Z}{|x|}\right)\ph_i(x) +2\sum_{k=1}^N\int\frac{|\ph_k(y)|^2\ph_i(x)-\ph_k(x)\overline{\ph_k(y)}\ph_i(y)}{|x-y|} dy = \eps_i\ph_i(x)
\end{equation}
for $N$ functions $\ph_1,\ldots,\ph_N\in L^2(\R^3\times\{1,\ldots,q\})$ subject to the constraints
\begin{equation}\label{onb}
  \int \overline{\ph_i(x)}\ph_j(x) dx =\delta_{ij}
\end{equation}
and real numbers $\eps_1,\ldots,\eps_N$. Here and henceforth
$x,y,\ldots$ denote elements $(\mathbf{x},s)$ of $\R^3\times\{1,\ldots,q\}$, $\int
dx$ denotes integration with respect
to the product of Lebesgue and counting measure, and $|x-y|=|\mathbf{x}-\mathbf{y}|$. Of course $q=2$ for
(spin-$1/2$) electrons but for later convenience we allow arbitrary
$q\in\N$. We have chosen atomic units where $\hbar=1$, the mass $m$ of the electron equals $1/2$ 
and the Rydberg energy equals $1$. 

The following theorem is our main result on critical points of the unrestricted 
Hartree-Fock functional. An analog result on a restricted Hartree-Fock
functional is given in Theorem~\ref{thm:rHF}. 

\begin{theorem}\label{main1}
Let $q,N,s$ and $n_1,\ldots,n_{s}$ be positive integers with
$1\leq n_1<n_2<\ldots<n_{s}$ and $ N = q\sum_{k=1}^{s}n_k^2$. Suppose
that 
\begin{equation}\label{Z-bound}
   Z> \frac{1}{\Delta_s}(20N+8\sqrt{2N}),\qquad \Delta_s=n_s^{-2}-(n_s+1)^{-2}.
\end{equation}
Then the Hartree-Fock equations \eqref{HFeq} have a solution $\ph_1,\ldots,\ph_N$ with 
\begin{equation}\label{ev-bounds}
   \eps_1,\ldots,\eps_N \in \cup_{k=1}^{s} \big[-Z^2/n_k^{2} , -Z^2/(n_k+1)^{2} \big),
\end{equation}
and the space spanned by $\ph_1,\ldots,\ph_N$ is uniquely determined by condition \eqref{ev-bounds}. The orbitals $\ph_i$ may be chosen 
in the form 
\begin{equation}\label{special-form}
    \ph_{nlm\sigma}(\mathbf{x},\mu) = \frac{f_{nl}(|\mathbf{x}|)}{|\mathbf{x}|}Y_{lm}(\mathbf{x})\delta_{\sigma,\mu}
\end{equation}
where $n\in\{n_1,\ldots,n_{s}\}$, $0\leq l \leq n-1$, $-l\leq m\leq l$, $\sigma\in\{1,\ldots,q\}$, and each of these quadruples 
$(n,l,m,\sigma)$ occurs exactly once.
\end{theorem}

The Hartree-Fock equations \eqref{HFeq} are equivalent to the Euler-Lagrange
equations of the Hartree-Fock functional
\begin{align}
\EE^{HF}(\ph_1,\ldots,\ph_N) =& \sum_{k=1}^{N} \int|\nabla\ph_k(x)|^2 - \frac{2Z}{|x|}|\ph_k(x)|^2\,dx\nonumber\\
    &+ 2\sum_{i<k}\int \frac{|\ph_i(x)|^2|\ph_k(y)|^2 - \overline{\ph_i(x)}\ph_i(y)\ph_k(x)\overline{\ph_k(y)}}{|x-y|}dxdy \label{HF-func}
\end{align}
where $\ph_1,\ldots,\ph_N \in L^2(\R^3\times\{1,\ldots,q\})$ are
subject to the constraints~\eqref{onb}. In fact, the Euler-Lagrange equations of \eqref{HF-func} are 
equations of the form \eqref{HFeq} with $\eps_i\ph_i$ replaced by the more general
term $\sum_{j=1}^N\lambda_{ij}\ph_j$,  the
coefficients $\lambda_{ij}$ being Lagrange multipliers. These generalized Hartree-Fock equations, as well as 
the Hartree-Fock functional and the constraints 
\eqref{onb} are invariant 
with respect to transformations $(\ph_1,\ldots,\ph_N)\mapsto (\tilde\ph_1,\ldots,\tilde\ph_N)$ of the form
$\tilde\ph_i = \sum_{j=1}^N U_{ij}\ph_j$
where $(U_{ij})$ denotes any unitary $N\times N$ matrix. By choosing
this matrix to diagonalize the self-adjoint matrix $(\lambda_{ij})$,
the equations \eqref{HFeq} emerge. This also means that the Hartree-Fock functional only depends on the $N$-dimensional 
subspace of $L^2(\R^3\times\{1,\ldots,q\})$ spanned by $\ph_1,\ldots,\ph_N$
or on the orthogonal projection 
\begin{equation}\label{proj}
P = \sum_{k=1}^N |\ph_k\rangle\langle \ph_k|
\end{equation}
onto this space. As a function of $P$ the Hartree-Fock functional is
quadratic and its domain can be extended to a convex set without
lowering the minimum \cite{Lieb1981, Bach1992}. Nevertheless the Hartree-Fock functional is not
convex due to the presence of the negative exchange term - the second
term in the numerator of \eqref{HF-func}. The convex
functional obtained by dropping the exchange term is called reduced
Hartree-Fock functional.

We are now in position to derive Theorem~\ref{main2}, below, on the
uniqueness of the minimizer of the Hartree-Fock functional \eqref{HF-func}. Suppose
we choose $n_1=1,n_2=2,\ldots$ in Theorem~\ref{main1}. Then, by
Proposition~\ref{eigenvalues}, for $Z$ sufficiently large
the condition \eqref{ev-bounds} becomes equivalent to the statement that 
$\eps_1,\ldots,\eps_N$ are the lowest $N$ eigenvalues of the Fock
operator, which is the linear operator acting on $\ph_i$ on the left hand side of \eqref{HFeq}. 
This condition on
$\eps_1,\ldots,\eps_N$ is satisfied for any solution of
\eqref{HFeq} associated with a minimizer of the Hartree-Fock
functional \cite{LS1977}. Hence the Theorem~\ref{main1} implies the following
theorem with the exception of the bound on $Z$.

%Like the Hartree-Fock functional the Fock operator only depends 
%on the projection \eqref{proj}.
%If we choose $n_1=1,n_2=2,\ldots$, then, for $Z$ sufficiently large,
%the condition \eqref{ev-bounds} becomes equivalent to the statement that 
%$\eps_1,\ldots,\eps_N$ are the lowest $N$ eigenvalues of the Fock operator. This condition is satisfied for any solution of
%\eqref{HFeq} associated with a minimizer of the Hartree-Fock functional \cite{LS1977}. Hence the Theorem~\ref{main1} implies:

\begin{theorem}\label{main2}
Let $s$, $q$ and $N$ be positive integers with
$N=q\sum_{n=1}^{s}n^2$ and suppose that
$Z>\Delta_s^{-1}(12N+4\sqrt{2N}-4)$ where $\Delta_s=s^{-2}-(s+1)^{-2}$.
Then the minimizer of the Hartree-Fock functional is unique in the
sense of a unique projection \eqref{proj}. It may be represented by $N$ orbitals $\ph_1,\ldots,\ph_N$
of the form \eqref{special-form} satisfying the Hartree-Fock equations \eqref{HFeq}. In particular, the density 
$\sum_{k=1}^N|\ph_k(x)|^2$ is unique and spherically symmetric.
\end{theorem}

For $N=q=2$ and $Z\geq 35$ the Theorem~\ref{main2} says that the Hartree-Fock functional has a unique 
minimizer given by two one-electron orbitals with opposite spins and equal spatial wave function $\ph\in L^2(\R^3)$, 
$\ph$ being the unique minimizer of the restricted Hartree functional 
$\ph\mapsto \EE^H(\ph,\ph)$, where $\EE^H$ is defined in
\eqref{H-func}, below.  In particular, if $q=2$ and $Z\geq 35$ then $\inf\EE^{HF}(\ph_1,\ph_2)\geq
\inf \EE^{H}(\ph,\ph)$. This statement is false for $Z<1.0268$ (see the remark after Theorem~\ref{thm:H}), which proves the necessity of some lower 
bound on $Z$ other than $Z>N-1$.

The (unrestricted) $N$-particle Hartree functional $\EE^{H}$ is defined on arbitrary $N$ tuples of functions 
$\ph_k\in H^1(\R^3)$ with $\int |\ph_k(x)|^2dx=1$, and it is given by 
\begin{align}
   \EE^H(\ph_1,\ldots,\ph_N) =& \sum_{k=1}^N \int |\ph_k(x)|^2 - \frac{2Z}{|x|}|\ph_k(x)|^2 dx\nonumber\\
 &+2\sum_{i<k}\int \frac{|\ph_i(x)|^2|\ph_k(y)|^2}{|x-y|}dxdy.\label{H-func}
\end{align}
No orthogonality is assumed on $\ph_1,\ldots,\ph_N$. It is well-known that $\EE^H$ 
has a minimizer if $Z>N-1$ and that minimizing orbitals are pointwise positive by a suitable choice of their phase \cite{LS1977}.
To establish uniqueness of the minimizer we consider $\EE^H$ as a spin-restricted Hartree-Fock functional with $q=N$, 
the spin-restriction being that $\phi_k(\mathbf{x},s)=\ph_k(\mathbf{x})\delta_{k,s}$. Then $\EE^{HF}(\phi_1,\ldots,\phi_N)=\EE^H(\ph_1,\ldots,\ph_N)$, 
and the following theorem, with the exception of the bound on $Z$, follows from Theorem~\ref{main2} with $q=N$.

\begin{theorem}\label{thm:H}
Let $N\in\N$ and suppose that $Z> 3^{-1}(40N+16\sqrt{2N}-8)$. Let 
$\ph_1,\ldots,\ph_N$ be any minimizer of the Hartree functional \eqref{H-func}. 
Then up to phases, $\ph_k=\ph$ for all $k$ where $\ph$ is the unique, 
positive minimizer of the restricted Hartree functional $\ph\mapsto \EE^{H}(\ph,\ldots,\ph)$.
\end{theorem}
%The following table gives integer values $Z_c$ for the lower bound on
%$Z$ from Theorem~\ref{thm:H}:
\begin{figure}[h]
\begin{center}
\begin{tabular}{|c|c|c|c|c|c|c|c|c|} 
\hline
$N$ & $2$  & $3$  & $4$  & $5$  & $6$  & $7$   & $8$   & $9$ \\ \hline
$Z_c$   & $35$ & $51$ & $66$ & $81$ & $96$ & $111$ & $126$ & $140$\\ \hline
\end{tabular}
\caption{The Hartree minimizer is unique for $Z\geq Z_c$.}
\end{center}
\end{figure}

In general, the minimum of $\EE^H$ is strictly below the minimum of
the restricted Hartree functional $\EE^{H}(\ph,\ph)$. 
In fact, Ruskai and Stillinger show that for $Z<1.0268$ and $N=2$ the restricted Hartree functional is bounded below by $-1$, while for 
$Z>1$ it is straightforward to show that the unrestricted Hartree
functional has its  minimum below $-1$ \cite{RuskStill}. A similar phenomenon of phase segregation 
is described in \cite{Aschbacher2009} for the two-electron Hartree functional with a confining external potential.

We now briefly sketch our strategy for proving Theorem~\ref{main1}. 
First, we use the well-known fact that $N$ functions
$Z^{3/2}\ph_i(Z\mathbf{x},\sigma)$, $i=1,\ldots,N$, $\ph_i\in
H^2(\R^3\times\{1,\ldots,q\})$, form a solution to the Hartree-Fock
equations \eqref{HFeq} if and only if $\ph_1,\ldots,\ph_N$ is a
solution to the new set of equations
\begin{equation}\label{HFeq2}
 \left(-\Delta -\frac{2}{|x|}\right)\ph_i(x) +
\frac{2}{Z}\sum_{k=1}^N\int\frac{|\ph_k(y)|^2\ph_i(x)-\ph_k(x)\overline{\ph_k(y)}\ph_i(y)}{|x-y|} dy = Z^{-2}\eps_i\ph_i(x).
\end{equation}
These rescaled Hartree-Fock equations clearly exhibit the perturbative
nature of the electron-electron repulsion in the large $Z$
limit. Let $H_P$ denote the rescaled Fock operator on the left hand
side in \eqref{HFeq2}, see \eqref{HF-operator}, the index $P$ being the 
projection onto $\spa\{\ph_1,\ldots,\ph_N\}$.
For $Z$ sufficiently large, $H_P$ has exactly $N$ eigenvalues,
counting multiplicities, near the eigenvalues $-n_k^{-2}$, $k=1,\ldots,s$, of $-\Delta-2/|x|$. 
More precisely, if $d$ is smaller than the gap $n_s^{-2}-(n_s+1)^{-2}$
and if $Z$ is sufficiently large, then the spectrum of $H_P$ in
$\Omega:=\cup_{k=1}^{s}[-n_k^{-2},-n_k^{-2}+d]$ consists of $N$ eigenvalues. 
The Hartree-Fock equations subject to \eqref{ev-bounds} are therefore equivalent to 
\begin{equation}\label{eq:fixpoint}
        P = \chi_{\Omega}\big(H_P\big).
\end{equation}
Here we use that all shells are filled. By making $Z$ even larger, if necessary, we can achieve that 
$P\mapsto \chi_{\Omega}\big(H_P\big)$ becomes a contraction, and hence that \eqref{eq:fixpoint} has a unique solution. 
This strategy obviously requires good eigenvalue estimates for $H_P$
and good control of $\chi_{\Omega}\big(H_P\big)-\chi_{\Omega}\big(H_Q\big)$ in terms of
$P-Q$. Concerning the second point we develop a novel method that yields much better bounds than, e.g.,
a resolvent integral for $\chi_{\Omega}\big(H_P\big)$ would give. 

Our work was inspired by the paper of Huber and Siedentop on solutions
of the Dirac-Fock equations  \cite{HubSied2007}. Using the contraction principle, they solve an equation 
analog to \eqref{eq:fixpoint} with the Laplacian replaced by the Dirac operator. 
We expect that our methods would allow to improve the results in \cite{HubSied2007}.
Minimization problems for semi-bounded Dirac-Fock type functionals are
studied in \cite{HLS}, and in the translation invariant case (no external potential) 
the minimizer is shown to be unique. When an external Coulomb
potential is present the method of \cite{HLS} does not seem to
work. For results on absence of a Hartree-Fock minimizer see \cite{Lieb1984,
Solovej2003}. Existence and uniqueness of radial solutions to Hartree equations are derived in \cite{Reeken1970}.
Uniqueness results concerning Hartree equations with \emph{attractive} Coulomb interactions
 are established, e.g., in \cite{Lenzmann2009, AFGST}. Several of the aforementioned results have been extended to pseudo-relativistic
Hartree-Fock functionals and to Hartree-Fock functionals including a
magnetic field \cite{AquaSolo2010, MelEns2008a, MelEns2008b,
  MelEns2009}. Last but not least we should mention the fundamental paper
of Bach on the accuracy of the Hartree-Fock approximation to
the quantum mechanical ground state energy \cite{Bach1992}.

This paper is organized as follows. In Section~2 we introduce all notations and we prove eigenvalue bounds for 
Fock operators. Section~3 contains the proofs of all theorem in this introduction.
In Section~4 we establish a theorem analog to Theorem~\ref{main1} in restricted HF-theory. Here, for proving existence of solutions 
we use the Schauder-Tychonoff theorem, which does not require as large
values of $Z$ as the contraction property does.
\\

\emph{Acknowledgement.} M.G. thanks Alex Elgart, George Hagedorn and
Timo Weidl for useful and inspiring discussions, Mathieu Lewin for explaining to him the uniqueness proof in \cite{HLS},
and I.M. Sigal for hospitality at the University of Toronto, where
part of this work was done. F.~H. is supported by the
\emph{Studienstiftung des Deutschen Volkes}.
\section{Notations and Eigenvalue Estimates}

In this section we collect the operator- and eigenvalue estimates
needed in later sections. We also introduce the definitions and notations used throughout the paper.

The sets of bounded linear operators, of Hilbert-Schmidt operators,
and of trace class operators in a separable Hilbert space $\HH$ are denoted by
$\LL(\HH)$, $\TT_2(\HH)$ and $\TT_1(\HH)$, respectively. The
corresponding norms are $\|\cdot\|$, $\|\cdot\|_2$, and
$\|\cdot\|_1$. Recall that 
$\TT_1(\HH)\subset \TT_2(\HH)\subset \LL(\HH)$ and that $\|K\|\leq
\|K\|_2\leq \|K\|_1$ for all $K\in \TT_1(\HH)$. Let 
\begin{equation*}
     S_{N,q} := \Big\{P\in \TT_1(L^2(\R^3 \times\{1,\dots,q\})) \Big|\, 0\leq P\leq 1,\ \tr P=N\Big\}.
\end{equation*}
The only reason for working with $S_{N,q}$ rather than with 
the set of self-adjoint projections of rank $N$, is that we need $S_{N,q}$ to be convex in Section~\ref{section:RHF}.

For each $P\in S_{N,q}$ there is a unique square-integrable kernel
$\tau=\tau_P$ such that $P\ph(x)=\int \tau(x,y)\ph(y) dy$ and there is a unique way to associate a 
density $\rho=\rho_P\in L^1$ with $P$, see Lemma~\ref{trace-formula}. If $P=\sum_n\lambda_n|\ph_n\rangle\langle\ph_n|$, then 
\begin{eqnarray}
\tau(x,y) & = & \sum_{n\geq 0} \lambda_n \ph_n(x)\overline{\ph_n(y)}, \label{def-tau} \\
\rho(x) & = & \sum_{n\geq 0}  \lambda_n| \ph_n(x)|^2 \label{def-rho}.
\end{eqnarray}
For given $P\in S_{N,q}$ we define a Fock operator
$H_P$ in $L^2(\R^3 \times\{1,\dots,q\})$ by 
\begin{equation}\label{HF-operator}
    H_P = -\Delta - V + \frac{1}{Z}(U_P - K_P),
\end{equation}
where $V$, $U_P$ are the multiplication operators associated with the
real-valued functions $V(x)=2/|x|$ and $U_P=\rho*V$, and for $\psi\in H^2$,
$$
  (K_P \psi)(x) := \int V(x-y)\tau(x,y)\psi(y)\,dy. \label{def K_P}
$$
From the fact that $V$ has a positive Fourier transform it is easy to see, using \eqref{def-tau}, 
that $K_P\geq 0$, and $U_P-K_P\geq 0$ by a straightforward computation. Hence $H_P\geq -\Delta-V\geq -1$.
Moreover, by the Kato-Rellich theorem, $H_P$ is self-adjoint 
on the Sobolev space $H^2(\R^3 \times\{1,\dots,q\})$.

By remarks in the introduction we may consider $\EE^{HF}$ as a function on $S_{N,q}\cap\{P^2=P\}$ given by
$\EE^{HF}(P)=\EE^{HF}(\ph_1,\ldots,\ph_N)$ where $(\ph_1,\ldots,\ph_N)$ is any orthonormal basis of $\ran P$.
Explicitely,
$$
  \EE^{HF}(P) = \tr\big[(-\Delta-V)P\big]+ \frac{1}{2Z}\int \big(\rho(x)\rho(y)-|P(x,y)|^2\big)V(x-y)\,dxdy.
$$

\begin{lemma}
\label{form-estimates}
Let $P \in S_{N,q}$. Then for all $\eps >0$
\begin{itemize}
\item[(i)] $\disp V \leq -\eps \Delta + \eps^{-1}$,
\item[(ii)] $\disp U_P \leq N\left( -\eps \Delta + \eps^{-1} \right)$,
\item[(iii)] $\disp\frac{1}{Z}U_P \leq \eps\left(-\Delta-V \right) + \frac{1}{\eps}\left(\eps+\frac{N}{Z}\right)^2$,
\item[(iv)] $\disp -\Delta \leq \frac{1}{1-\eps}H_P + \frac{1}{\eps(1-\eps)}\quad\text{if}\quad\eps\in (0,1)$.  
\end{itemize}
\end{lemma}

\begin{proof}
Statement (i) follows from $-\Delta-V \geq -1$ and from the scaling
properties of $-\Delta$ and $V$ with respect to the unitary transformation
$\ph(\mathbf{x},\sigma) \mapsto \eps^{-3/2}\ph(\mathbf{x}/\eps,\sigma)$. Statement (ii) follows from $U_P=\rho*V$, from (i) and from
$\int\rho(x) dx=N$. 

To prove (iii), fix $\eps > 0$. By (i) and (ii),
$$
 \eps V + \frac{1}{Z} U_P \leq \left( \eps + \frac{N}{Z} \right) \left( -\delta\Delta+\delta^{-1} \right)
$$
for all $\delta >0$. Upon subtracting $\eps V$ from both sides and
making the choice $\delta\left(\eps+\frac{N}{Z}\right) = \eps$ for
$\delta$, the desired estimate follows. Inequality (iv) follows from $-\Delta \leq H_P +V$ and from (i).
\end{proof}

\begin{proposition}\label{eigenvalues}
Given $N\in\N,\ Z > 0$ and $P \in S_{N,q}$, let $E_n^\infty$ and $E_n^Z$ denote the $n$-th eigenvalue, counting multiplicities, of the
Schrödinger operator $-\Delta-V$ and the Hartree-Fock operator
$H_P$, respectively. Then:
\begin{itemize}
\item[(a)] For all $n\in\N$, $E_n^\infty \leq E_n^Z$ and
$$
E_n^Z \leq E_n^\infty + 2\frac{N}{Z} + 2\frac{N}{Z}\sqrt{E_n^{\infty}+1}.
$$
\item[(b)] If $P$ minimizes the Hartree-Fock functional, then $E_N^Z$ obeys the
  following estimate: 
$$
E_N^Z \leq E_N^\infty + 2\frac{N-1}{Z}+2\frac{N-1}{Z}\sqrt{E_N^\infty+1}.
$$
\end{itemize}
\end{proposition}

\begin{proof}
From $U_P-K_P \geq 0$, $K_P \geq 0$ and Lemma~\ref{form-estimates}~(iii) we see that
$$
-\Delta-V \leq H_P \leq -\Delta-V  + \frac{1}{Z}U_P 
\leq (1+\eps)\left(-\Delta-V\right)+\frac{1}{\eps}\left(\eps +\frac{N}{Z}\right)^2
$$
for all $\eps >0$. By the min-max principle, this implies that
$$
E_n^\infty \leq E_n^Z \leq \left( 1 + \eps \right)E_n^\infty + \frac{1}{\eps} \left( \eps+\frac{N}{Z}\right)^2.
$$
Optimizing with respect to $\eps$ yields the desired estimates of part (a).

To prove (b) let $\ph_1,\ldots,\ph_N$ be an orthonormal basis of $\ran P$ with $H_P\ph_k=E_k^Z\ph_k$. 
Then the Hartree-Fock functional can be decomposed as
\begin{eqnarray*}
\EE^{HF}(\ph_1,\dots,\ph_N) & = & \EE^{HF}_{N-1}(\ph_1,\dots,\ph_{N-1}) + \sprod{\ph_N}{H_{P,N-1}\ph_N} \\
& = & \EE^{HF}_{N-1} (\ph_1,\dots,\ph_{N-1}) + E_N^Z,
\end{eqnarray*}
where $\EE^{HF}_{N-1}$ and $H_{P,N-1}$ denote the Hartree-Fock
functional and Fock operator belonging to the $(N-1)$-particle integral
kernel $\tau_{N-1}(x,y) = \sum_{k=1}^{N-1}
\ph_k(x)\overline{\ph_k(y)}$. Since $\psi \mapsto \EE^{HF}(\ph_1, \dots, \ph_{N-1}, \psi)$ is minimized by $\ph_N$,
$$
E_N^Z = \inf_{\begin{array}{l} \scriptstyle \psi \in D(H_{P,N-1}); \ \norm{\scriptstyle\psi}=1 \\ \scriptstyle \psi \in \spa\{\ph_1,\dots,\ph_{N-1}\}^\perp \end{array}} \sprod{\psi}{H_{P,N-1}\psi}.
$$
Thus, using the min-max principle again, $E_N^Z$ is bounded from above
by the $N$-th eigenvalue of $H_{P,N-1}$. Part (a) applied to $H_{P,N-1}$ completes the proof of (b).
\end{proof}

%------------------------------------------------------------------------------------------------------------------------------------------

%\newpage
\section{Solving the Hartree-Fock Equations}

This section contains the proofs of all the theorems given in the
introduction. For restricted Hartree-Fock theory see the next
section. Our main tool is the following abstract result comparing the
spectral projections $\chi_\Omega(A)$ and $\chi_\Omega(B)$ of two
self-adjoint operators $A$ and $B$.

\begin{proposition}\label{lm-abstract}
Let $A,B:D\subset\HH \to \HH$ be self-adjoint operators and let
$\Omega\subset\R$ be a bounded Borel set for which the spectra of $A$ and $B$
satisfy the gap conditions
\begin{equation}\label{gap-assumption}
\begin{split}
\dist(\sigma(A)\cap\Omega, \sigma(B)\setminus\Omega) &\geq \delta,\\
\dist(\sigma(B)\cap\Omega, \sigma(A)\setminus\Omega) &\geq \delta,
\end{split}
\end{equation}
for some $\delta>0$. Suppose $A$ and $B$ have only point spectrum in $\Omega$. Then 
$$
\norm{\chi_\Omega(A)-\chi_\Omega(B)}_2 \leq \delta^{-1} \left( \norm{(A-B)\chi_\Omega(A)}_2^2 + \norm{(A-B)\chi_\Omega(B)}_2^2 \right)^{1/2}.
$$
\end{proposition}

\begin{proof}
We only prove the proposition in the case where $\chi_\Omega(A)$ and $\chi_\Omega(B)$ are finite rank projections. 
The more general case of infinite point spectrum in $\Omega$ is left as an exercise for the reader.
Let $\chi(A) := \chi_\Omega(A)$ and $\chi(B) := \chi_\Omega(B)$ for short. Then
$\chi(A)^2=\chi(A)=\chi(A)^*$ and similarly for $\chi(B)$. It follows that
\begin{align}
\norm{\chi(A)-\chi(B)}^2_2 =&\tr(\chi(A)(1-\chi(B))\chi(A))\nonumber\\
&+ \tr(\chi(B)(1-\chi(A))\chi(B))\label{cyclic-trace-eqn}
\end{align}
where we also used the cyclicity of the trace. To estimate the first
term on the right hand side of \eqref{cyclic-trace-eqn} we choose an
orthonormal basis $(\ph_k)_{k=1}^n$ of $\chi(A)$ consisting of eigenfunctions of $A$:
$$
A\ph_k = \eps_k \ph_k, \quad \eps_k \in \sigma(A)\cap\Omega, \ k=1,\dots,n.
$$
By the gap assumption \eqref{gap-assumption}, $|\lambda-\eps_k| \geq \delta$ for $\lambda \in \sigma(B) \setminus \Omega$ and 
hence, by the spectral theorem,
$$
\frac{1}{\delta} | B - \eps_k | \geq 1 - \chi_\Omega(B), \quad k=1,\dots,n.
$$
We conclude that
\begin{eqnarray}
\tr(\chi(A)(1-\chi(B))\chi(A)) & = & \sum_{k=1}^n \sprod{\ph_k}{(1-\chi(B))\ph_k} \nonumber \\
& \leq & \delta^{-2} \sum_{k=1}^n \sprod{\ph_k}{(B-\eps_k)^2\ph_k} \nonumber \\
& = & \delta^{-2} \sum_{k=1}^n \sprod{\ph_k}{(B-A)^2\ph_k} \nonumber \\
& = & \delta^{-2} \norm{(A-B)\chi(A)}^2_2. \label{trace estimate}
\end{eqnarray}
The proposition follows from \eqref{cyclic-trace-eqn}, \eqref{trace
  estimate} and from an estimate similar to \eqref{trace estimate} with $A$ and $B$ interchanged.
\end{proof}

\begin{lemma}
\label{H1-estimates}
For all $\ph\in H^1(\R^3\times\{1,\dots,q\})$ and $P,Q \in S_{N,q}$
\begin{itemize}
\item[(i)] $\norm{(U_P-U_Q)\ph} \leq 4 \norm{P-Q}_1 \norm{\nabla \ph}$
\item[(ii)] $\norm{(K_P-K_Q)\ph} \leq 4 \norm{P-Q}_2 \norm{\nabla \ph}.$
\end{itemize}
\end{lemma}

\begin{proof}
For all $\psi \in L^2(\R^3 \times \{1,\dots,q\} )$
\begin{eqnarray*}
|\sprod{\psi}{(U_P-U_Q)\ph}| & \leq & 2 \int \,dy |\rho_P(y)-\rho_Q(y)| \int |\psi(x)| \frac{|\ph(x)|}{|x-y|} \,dx \\
& \leq & 4 \int \,dy |\rho_P(y)-\rho_Q(y)| \norm{\psi} \norm{\nabla \ph} \\
& \leq & 4 \norm{P-Q}_1 \norm{\psi} \norm{\nabla \ph}
\end{eqnarray*}
by Cauchy-Schwarz, the uncertainty principle lemma \cite[section
X.2]{ReedSimon2} and by Lemma~\ref{trace-formula}. This proves (i). The proof of (ii) is similar.
\end{proof}

\begin{proposition}\label{contraction}
Let $P,Q \in S_{N,q}\cap\{P^2=P\}$ and let $\Omega \subset (-\infty,0)$ be a bounded Borel set such that
$\dist (\Omega,\sigma(H_P)\setminus \Omega) \geq \delta$ and $\dist (\Omega,\sigma(H_Q)\setminus \Omega) \geq \delta$
for some $\delta >0$. Suppose moreover that
$\tr(\chi_{\Omega}(H_P))=N=\tr(\chi_{\Omega}(H_Q))$. Then
\begin{eqnarray*}
\lefteqn{\norm{\chi_\Omega(H_P)-\chi_\Omega(H_Q)}_2} \\
&& \leq\frac{4}{\delta Z} \left( 1 + \sqrt{2N} \right)
\left(\norm{\sqrt{-\Delta}\chi_\Omega(H_P)}^2_2 +
  \norm{\sqrt{-\Delta}\chi_\Omega(H_Q)}^2_2 \right)^{1/2}
\norm{P-Q}_2\\
&& \leq  \frac{8}{\delta Z}\left(1+\sqrt{2N}\right)\sqrt{2N}\norm{P-Q}_2.
\end{eqnarray*}
The factor 8 in the last line may be replaced by 4 if both $P$ and $Q$
satisfy the Hartree-Fock equation \eqref{eq:fixpoint}. It may be replaced by $4\sqrt{2.5}$
if $P$ or $Q$ satisfies \eqref{eq:fixpoint}. 
\end{proposition}

\emph{Remark.}
The set $S_{N,q}\cap\{P^2=P\}$ is a closed subset of the
Hilbert-Schmidt operators on $L^2(\R^3\times\{1,\ldots,q\})$ and hence
it is complete with respect to the metric $d(P,\tilde P) = \|P-\tilde
P\|_2$. 

\begin{proof}
Applying Proposition~\ref{lm-abstract} to $H_P$ and $H_Q$ we see that we
need to estimate $\norm{(H_P-H_Q)\chi_{\Omega}(H_P)}_2$ and the same
expression with $P$ and $Q$ interchanged. By definition of $H_P$ and
$H_Q$, 
$$
   H_P-H_Q = Z^{-1}(U_P-U_Q) - Z^{-1}(K_P-K_Q)
$$and by Lemma~\ref{H1-estimates},
\begin{align*}
    \|(U_P-U_Q)\chi_\Omega(H_P)\|_2 &\leq 4\|P-Q\|_1 \norm{\sqrt{-\Delta}\chi_\Omega(H_P)}_2\\
     \|(K_P-K_Q)\chi_\Omega(H_P)\|_2 &\leq 4\|P-Q\|_2 \norm{\sqrt{-\Delta}\chi_\Omega(H_P)}_2.
\end{align*}
The first inequality of the theorem now follows from $\norm{P-Q}_1 \leq \sqrt{2N}
\norm{P-Q}_2$. Here we use that $P$ and $Q$ have rank $N$. From Lemma~\ref{form-estimates} (iv) with $\eps=1/2$ we see that 
$-\Delta\leq 2H_P+4$. Therefore
\begin{align*}
  \norm{\sqrt{-\Delta}\chi_\Omega(H_P)}_2^2 
&= \tr \big(\chi_\Omega(H_P)(-\Delta)\chi_\Omega(H_P)\big)\\
&\leq 4\tr \chi_\Omega(H_P) =4N. 
\end{align*}
If $P$ solves the Hartree-Fock equation then $
\norm{\sqrt{-\Delta}\chi_\Omega(H_P)}_2^2=|\EE^{HF}(P)|\leq N$ by the virial
theorem and because $-\Delta-V\geq -1$.
\end{proof}

For each $R\in\text{SO}(3)$ we define a unitary operator $U(R)$ in $L^2(\R^3)$ by
\begin{equation}\label{rotation}
[U(R)\psi](x):=\psi(R^{-1}x).
\end{equation}

The following theorem describes the content of Theorem~\ref{main1} in terms of rank $N$ projections.

\begin{theorem}\label{thm:uni-pro}
Under the assumptions of Theorem~\ref{main1} the Hartree-Fock equation \eqref{eq:fixpoint} 
has a unique solution $P\in S_{N,q}$ with the property that $Z^2\sigma(H_P \restricted P\HH)$ is given by \eqref{ev-bounds}. 
This $P$ is of the form $P=P' \otimes 1$ with respect to $L^2(\R^3) \otimes \C^q$ where $P'\in S_{N/q,1}$ 
and moreover $P'=U(R)P'U(R)^*$ for all $R\in\text{SO}(3)$.
\end{theorem}

\begin{proof}
We first prove existence and uniqueness of $P$ using Proposition~\ref{contraction} with
$$
\Omega := \bigcup_{k=1}^{s} \left[ -n_k^{-2}, -n_k^{-2} + \frac{4N}{Z} \right].
$$
By Proposition~\ref{eigenvalues}, for $Z$ large enough and all $P\in S_{N,q}$
$$
\dist(\Omega, \sigma(H_P) \setminus \Omega) \geq n_{s}^{-2} - (n_{s}+1)^{-2} - \frac{4N}{Z} =: \delta,
$$
$\delta>0$, and $\tr \chi_{\Omega}(H_P)=N$. Hence the Proposition~\ref{contraction}, the remark thereafter, and the
contraction principle imply that the equation $P=\chi_\Omega(H_P)$ has a unique solution $P$ provided that
\begin{equation} \label{usc}
\frac{8}{Z} \left( 1+ \sqrt{2N} \right) \sqrt{2N} < \delta.
\end{equation}
This is satisfied for $Z$ obeying \eqref{Z-bound}.

We next show that $P$ is of the form $P'\otimes 1$ in $L^2(\R^3) \otimes \C^q$ where $P' \in S_{N/q,1}$. 
To this end we consider the modified Hartree-Fock equation 
\begin{equation}
\label{mod-HFE}
P' = \chi_\Omega(H_{P'}^{(q)})
\end{equation}
where 
$$
H_{P'}^{(q)} := -\Delta -V + \frac{q}{Z} U_{P'} - \frac{1}{Z} K_{P'}
$$
in $L^2(\R^3)$. The eigenvalues of $H_{P'}^{(q)}$ satisfy the estimate given by Proposition~\ref{eigenvalues} a). 
Therefore, the arguments above show that \eqref{mod-HFE} has a unique
solution $P'$ because \eqref{usc} holds by assumption \eqref{Z-bound}. 
From  $H_{P'\otimes 1}=H_{P'}^{(q)}\otimes 1$ and \eqref{mod-HFE}
it follows that $P' \otimes 1$ solves the Hartree-Fock equation \eqref{eq:fixpoint}. 
Hence, $P=P' \otimes 1$ by the uniqueness of the solution to \eqref{eq:fixpoint}.

Finally we prove that $P'$ commutes with $U(R)$. From the spherical symmetry of $V$ it follows that $U(R) H_{P'}^{(q)} U(R)^* = H^{(q)}_{P'(R)}$ where $P'(R) = U(R) P' U(R)^*$. 
Using \eqref{mod-HFE}, we conclude that
$$
P'(R) = U(R) \chi_\Omega(H^{(q)}_{P'}) U(R)^* = \chi_\Omega\Big(U(R) H_{P'}^{(q)} U(R)^*\Big) = \chi_\Omega\Big(H^{(q)}_{P'(R)}\Big)
$$
which implies $P'(R)=P'$ because the solution to \eqref{mod-HFE} is unique.
\end{proof}

\begin{proof}[\textbf{Proof of Theorem~\ref{main1}}]
Theorem~\ref{thm:uni-pro} tells us that $P=P'\otimes 1$ where $P'$ commutes with all rotations. 
Moreover, $H_P=H_{P'}^{(q)}\otimes 1$,
$P'=\chi_{\Omega}(H_{P'}^{(q)})$ and $H_{P'}^{(q)}$ commutes with all
rotations as well. Let $\eps_k\in \Omega$ be an eigenvalue of $H_{P'}^{(q)}$. 
The eigenspace associated with $\eps_k$ carries a representation of
$SO(3)$ given by \eqref{rotation}. Its irreducible subspaces are spanned by functions of the form 
\begin{equation}\label{ylm}
|\mathbf{x}|^{-1}f(|\mathbf{x}|)Y_{lm}(\mathbf{x}),\qquad f\in L^2(\R_{+}),
\end{equation}
where $Y_{lm}$ denotes a spherical harmonic. 
Now fix $l$ and $m$ and let $\HH_{lm}$ denote the space of all functions of the form \eqref{ylm} with 
arbitrary $f$. This space is reducing for both $-\Delta-V$ and $H_{P'}^{(q)}$. Since the spectrum of $(-\Delta-V)\upharpoonright\HH_{lm}$ in 
$(-\infty,0)$ consists of the simple eigenvalues $-1/n^2$, $n\geq l+1$,
it follows from Proposition~\ref{eigenvalues} and from the assumption
on $Z$ that $H_{P'}^{(q)}\upharpoonright\HH_{lm}$ has exactly one
eigenvalue in each of the intervals
$[-n_{k}^{-2},-(n_k+1)^{-2})$ with $n_k\geq l+1$. This completes the proof of Theorem~\ref{main1}.
\end{proof}

\begin{proof}[\textbf{Proof of Theorem~\ref{main2}}]
For $Z>N-1$ the Hartree-Fock functional is known to have a minimizer
and any minimizer $P$ is the spectral projection onto the spectral
subspace of $H_P$ associated with the lowest $N$ eigenvalues \cite{LS1977}. If
$N=q\sum_{n=1}^{s}n^2$ for some $s$, then, by Proposition~\ref{eigenvalues} the lowest $N$
eigenvalues of $H_P$ belong to $[-1,-(s+1)^{-2})$ provided that
$4(N-1)/Z<s^{-2}-(s+1)^{-2}$. Thus for sufficiently large $Z$ Theorem~\ref{main1} implies uniqueness of
the minimizer of the Hartree-Fock functional as well as the assertions
on the one-particle orbitals. The bound on $Z$ is obtained by
inspection of the proof of Theorem~\ref{main1} keeping in mind that
$P$ is a minimizer which is given. Hence the improved
bounds from Proposition~\ref{eigenvalues} and Proposition~\ref{contraction} are available.
\end{proof}

\begin{proof}[\textbf{Proof of Theorem~\ref{thm:H}}]
Given $\ph_1,\dots,\ph_N \in H^1(\R^3)$ with $\int |\ph_k(x)|^2 \,dx = 1$ let $\Phi_1,\dots,\Phi_N \in H^1(\R^3\times\{1,\ldots,N\})$ be defined by
\begin{equation}
\label{Hartree-eqn1}
\Phi_k(\mathbf{x},s) := \ph_k(\mathbf{x})\delta_{ks}, \quad k=1,\dots,N.
\end{equation}
Then $\Phi_1,\dots,\Phi_N$ are orthonormal in $L^2(\R^3\times\{1,\ldots,N\})$ and
\begin{equation}
\label{Hartree-eqn2}
\EE^{HF}(\Phi_1,\dots,\Phi_N) = \EE^{H} (\ph_1,\dots,\ph_N)
\end{equation}
by the definitions of $\EE^{HF}$ and $\EE^H$.  By \eqref{Hartree-eqn2}, the minimization problem for $\EE^{H}$ 
is equivalent to the minimization problem 
for $\EE^{HF}$ with $q=N$ in the restricted class of orbitals of the form \eqref{Hartree-eqn1}. 
By Theorem~\ref{main2}, the (unrestricted) Hartree-Fock functional for
$q=N$ and $Z$ sufficiently large has a unique minimizer 
$P\in S_{N,N}$, which is of the form $P=P'\otimes 1$ with respect to $L^2(\R^3\times\{1,\ldots,N\})=L^2(\R^3)\otimes\C^N$. 
$P'$ here is a rank one projection that commutes with rotations. Hence $(P'\psi)=\ph\sprod{\ph}{\psi}$ with a spherically symmetric function $\ph$. 
It follows that $\EE^{HF}(P)=\EE^{H}(\ph,\ldots,\ph)$. It remains to
determine the condition on $Z$ for uniqueness of the Hartree-Fock minimizer in the present case where $q=N$. 
To this end we use the improved eigenvalue estimate from Proposition~\ref{eigenvalues} b). The gap condition becomes
\begin{equation} \label{gap condition}
\delta := \frac{3}{4} - 2\frac{N-1}{Z} > 0
\end{equation}
and the contraction condition reads
\begin{equation} \label{contraction condition}
\frac{4}{Z} \left( 1 + \sqrt{2N} \right) \sqrt{2N} < \delta.
\end{equation}
Both, (\ref{gap condition}) and (\ref{contraction condition}) are satisfied if
$$
Z > \frac{40}{3}N + \frac{16}{3} \sqrt{2N} - \frac{8}{3}.
$$
\end{proof}

\begin{corollary}
\label{orbital convergence}
Under the assumptions of Theorem~\ref{thm:uni-pro} let $P_Z$ be the unique solution provided by this theorem. Then
\begin{equation}
\label{projection-convergence}
\lim_{Z \to \infty} P_Z = \sum_{k=1}^{s} \chi_{\left\{-n_k^{-2}\right\}} (-\Delta-V).
\end{equation}
\end{corollary}

\begin{proof}
Let $P_\infty$ denote the right hand side of
\eqref{projection-convergence}. A copy of the proof of Proposition~\ref{contraction} with $P=P_Z$ and $H_Q$ replaced by $H_\infty = -\Delta-V$ shows that
$$
\norm{P_Z-P_\infty}_2 \leq \frac{4}{\delta Z} \left( 1 + \sqrt{2N} \right) \sqrt{2N} \norm{P_Z}_2 = \frac{4\sqrt{2}}{\delta} \frac{N}{Z} \left( 1 + \sqrt{2N} \right) \to 0 \quad (Z\to\infty).
$$
Here, by Proposition~\ref{eigenvalues}, $\delta$ can be chosen
independently of $Z$ for $Z$ sufficiently large.
\end{proof}

%-----------------------------------------------------------------------------------------------------------------------------------------------------
%-----------------------------------------------------------------------------------------------------------------------------------------------------
%
\section{Restricted Hartree-Fock Theory}
\label{section:RHF}

In this section the Hartree-Fock functional is restricted to
one-particle orbitals of the special form  \eqref{special-form}. 
This will allow for all the electron configurations found in noble gas atoms.  
In this section we set $q=1$ to simplify the presentation.

By a shell index we mean a pair of integers $(n,\ell)$ with $n\geq 1$ and $0\leq \ell \leq n-1$.

\begin{theorem}\label{thm:rHF}
Let $(n_1,\ell_1),\ldots,(n_s,\ell_s)$ be given, pairwise distinct shell indices with
$n_1\leq n_2\leq \ldots \leq n_s$. Suppose $N=\sum_{j=1}^s(2\ell_j+1)$
and let $\Delta_s:=n_s^{-2}-(n_s+1)^{-2}$. 
\begin{itemize}
\item[(i)] If $Z>4N/\Delta_s $ then there exist normalized functions $f_1,\ldots,f_s\in L^2(\R_{+})$ such that the $N$ functions
$$
      \ph_{jm}(x):= \frac{1}{|x|}f_j(|x|)Y_{\ell_j m}(x),\qquad j=1,\ldots,s,\ m=-\ell_j\ldots \ell_j,
$$
solve the Hartree-Fock equations \eqref{HFeq} with eigenvalues $\eps_j$ satisfying 
\begin{equation}\label{ev-Rbounds}
     -\frac{1}{n_j^2} \leq  \eps_j\leq -\frac{1}{n_j^2} + \frac{4N}{Z}.
\end{equation}
\item[(ii)] If $Z$ satisfies \eqref{Z-bound} then the functions $f_j$ in (i) are unique up to global phases.
 \end{itemize}
\end{theorem}

\emph{Remark:} A result similar to part (i) of this theorem with the
weaker assumption $Z\geq N$ is described in
Section~III.3 of \cite{Lions1987}. However, Lions' argument is based
on the unproven assertion that all eigenvalues of a radial Hartree-Fock
operator are simple.

\medskip
To prove Theorem~\ref{thm:rHF} we solve the fixed point equation 
\begin{equation}\label{fp-sym}
     P=F(P):= \sum_{j=1}^s \chi_{\Omega_j}(H_P)\pi_{\ell_j}
\end{equation}
where $\pi_{\ell}$ denotes the orthogonal projection associated with
the eigenvalue $\ell(\ell+1)$ of the square of the total angular
momentum operator, and 
$$
     \Omega_{j}:=\Big[-\frac{1}{n_j^2},-\frac{1}{n_j^2}+\frac{4N}{Z}\Big].
$$
The spherical symmetry will be imposed by restricting $F$ to the subset 
$$
      S_N^{sym}:= \big\{P\in S_{N,1} \big| U(R)PU(R)^{*}=P,\ \text{for all}\ R\in SO(3)\big\}
$$
of $S_{N,1}$. To find a solution of \eqref{fp-sym} we use the Schauder-Tychonoff theorem. Its uniqueness will 
follow from the contraction principle.

\begin{lemma}\label{r-contraction}
Suppose the hypotheses of Theorem~\ref{thm:rHF} are satisfied and that
$\delta:=\Delta_s-4N/Z>0$. Then $F(S_N^{sym})\subset
S_N^{sym}$ and for all $P,Q\in  S_N^{sym}$,
\begin{equation}
   \|F(P)-F(Q)\|_1 \leq 32\frac{N}{\delta Z}\|P-Q\|_1.\label{r-cont1}
\end{equation}
If $P,Q\in  S_N^{sym}\cap\{P^2=P\}$ then
\begin{equation}
   \|F(P)-F(Q)\|_2 \leq \frac{8}{\delta
     Z}\sqrt{2N}(1+\sqrt{2N})\|P-Q\|_2.\label{r-cont2}
\end{equation}
\end{lemma}

\emph{Remark.} The Lipshitz constant in \eqref{r-cont2}
agrees with the one in Proposition~\ref{contraction}.

\begin{proof}
We first show that $F( S_N^{sym})\subset  S_N^{sym}$. For $P\in S_N^{sym}$, the Fock operator $H_P$ commutes with all rotations $U(R)$  
and hence so does its spectral projection 
$\chi_{\Omega_j}(H_P)$. It follows that $F_j(P):=\chi_{\Omega_j}(H_P)\pi_{\ell_j}=\pi_{\ell_j}\chi_{\Omega_j}(H_P)$ is an orthogonal projection and that 
\begin{equation}\label{ons}
    F_j(P)F_k(P) = \delta_{jk}F_{j}(P) 
\end{equation}
because $(n_j,\ell_j)\neq (n_k,\ell_k)$ for $j\neq k$. Hence $F(P)=\sum_{j=1}^s F_{j}(P)$ is an orthogonal projection that commutes with all rotations. 
As for the trace of $F(P)$, we note that $H_P\upharpoonright
\pi_{\ell_j}\HH$ has exactly $2\ell_j+1$ eigenvalues in $\Omega_j$, counted with multiplicities. 
This shows that $\tr F(P)=\sum_{j=1}^s \tr F_j(P) =\sum_{j=1}^s (2\ell_j+1)=N$. 

By inspection of the proof of Proposition~\ref{lm-abstract}, using \eqref{ons} we see that 
\begin{equation}\label{rc1}
    \|F(P)-F(Q)\|_2^2 = \tr F(P)(1-F(Q))F(P) + \tr F(Q)(1-F(P))F(Q)
\end{equation}
where 
\begin{align}
    \tr F(P)(1-F(Q))F(P) &= \sum_{j=1}^s \tr F_j(P)(1-F(Q))F_j(P)\nonumber \\
     &\leq \delta^{-2}\sum_{j=1}^s \|(H_P-H_Q)F_j(P)\|_2^2.\label{rc2}
\end{align}
Here it is important that $H_P$ and $H_Q$ are considered as operators on $\pi_{\ell_j}L^2(\R^3)$, 
which have $2\ell_j+1$ eigenvalues in $\Omega_j$ separated by a gap of size $\delta$ from the rest of the spectrum. 
To bound $\|(H_P-H_Q)F_j(P)\|_2$ we use estimates from the proof of
Proposition~\ref{contraction} as well as $\|P-Q\|_2\leq \|P-Q\|_1$. We find
\begin{equation}\label{rc3}
     \|(H_P-H_Q)F_j(P)\|_2 \leq \frac{16}{Z} \sqrt{2\ell_j+1}\|P-Q\|_1.
\end{equation}
The proof of \eqref{r-cont1} is now completed by combining 
\eqref{rc1}, \eqref{rc2}, \eqref{rc3} with 
\begin{equation}\label{rc4}
    \|F(P)-F(Q)\|_1 \leq \sqrt{2N} \|F(P)-F(Q)\|_2
\end{equation}
and with estimates similar to \eqref{rc2}, \eqref{rc3} where the roles
of $P$ and $Q$ are interchanged. \eqref{rc4} follows from the fact
that $F(P)$ and $F(Q)$ have rank $N$. For proving \eqref{r-cont2} we
use  $\|P-Q\|_1\leq \sqrt{2N}\|P-Q\|_2$ instead of \eqref{rc4}. The
rest is similar to the proof of \eqref{r-cont1}.
\end{proof}

\begin{proof}[\textbf{Proof of Theorem~\ref{thm:rHF}}]
To prove (i) we check that the map $F:S_N^{sym}\to S_N^{sym}$ satisfies the hypotheses of the Schauder-Tychonoff theorem.
By construction, $S_N^{sym}$ is convex and closed in
$\TT_1(L^2(\R^3))$, and $F:S_N^{sym}\to S_N^{sym}$ is continuous by Lemma~\ref{r-contraction}.
To prove compactness of the map $F$ let $(P_n)$ be any sequence in $ S_N^{sym}$. Then $F(P_n)$ may be written in the form
$$
  F(P_n)=\sum_{k=1}^N |\ph_k^{(n)}\rangle\langle\ph_k^{(n)}|,\qquad H_{P_n}\ph_k^{(n)} = \eps_k^{(n)}\ph_k^{(n)}
$$
with eigenvalues $\eps_k^{(n)} <-(n_s+1)^{-2}$ for all $k,n$. By Lemma~\ref{form-estimates}, the sequences $(\ph_k^{(n)})_n$ are bounded in $H^1(\R^3)$.
We may assume they are weakly convergent in $H^1$ and hence locally convergent after   
passing to a subsequence of $(P_n)$. Hence $F(P_n)$ will be convergent 
provided that $\ph_k^{(n)}(x)\to 0$ as $|x|\to\infty$ uniformly in $n$. 
To prove this uniform decay we pick a function $\chi\in C^{\infty}(\R;[0,1])$ with $\chi(t)=0$ 
for $t\leq 1$ and $\chi(t)=1$ for $t\geq 2$. Let $\chi_R(x)=\chi(|x|/R)$ for $R>0$. The IMS-formula tells us that 
$$
   2\chi_R(H_{P_n}-\eps_k^{(n)})\chi_R = \chi_R^2(H_{P_n}-\eps_k^{(n)}) + (H_{P_n}-\eps_k^{(n)})\chi_R^2 - [[H_{P_n},\chi_R],\chi_R],
$$
where $ [[H_{P_n},\chi_R],\chi_R]= 2|\nabla\chi_R|^2 - Z^{-1}[[K_{P_n},\chi_R],\chi_R]$ and 
$$
    \chi_R(H_{P_n}-\eps_k^{(n)})\chi_R \geq (-2R^{-1}-\eps_k^{(n)})\chi_R^2 \geq c \chi_R^2
$$
for some $c>0$ if $R\geq R_0$ and $R_0$ is large enough. It follows that 
\begin{align*}
   c\|\chi_R \ph_k^{(n)}\|^2 &\leq \frac{1}{2Z}\Big|\sprod{\ph_k^{(n)}}{[\chi_R,[\chi_R,K_{P_n}]]\ph_k^{(n)}}\Big| + O(R^{-2})\\
   &= O(R^{-1}) ,\qquad (R\to\infty),
\end{align*}
uniformly in $n$.
In the last step we expanded the double commutator, we used $|\tau^{(n)}(x,y)|\leq \rho^{(n)}(x)^{1/2}\rho^{(n)}(y)^{1/2}$, 
Cauchy-Schwarz, the spherical symmetry of $\rho^{(n)}$ and 
Newton's theorem to replace $|x-y|^{-1}$ by $\min
(|x|^{-1},|y|^{-1})$. Here $\tau^{(n)}$ and $\rho^{(n)}$ denote the
kernel and density of $P_n$. Statement (ii) in Theorem~\ref{thm:rHF}
follows from \eqref{r-cont2} and the remark following Lemma~\ref{r-contraction}.
\end{proof}

%-----------------------------------------------------------------------------------------------------

\section{Appendix}

The following lemma seems to be well-known, but we have not found it in the literature.

\begin{lemma}
\label{trace-formula}
There exists a bounded linear mapping
\begin{eqnarray*}
\rho : \TT_1(L^2(\R^n)) & \to & L^1(\R^n) \\
K & \mapsto & \rho_K
\end{eqnarray*}
which is uniquely determined by
$$
|\psi \rangle \langle \ph | \mapsto \psi(x)\overline{\ph(x)}.
$$
Furthermore, $\rho$ has the following properties:
\begin{itemize}
\item[(i)] $\int\rho_K(x) \,dx = \tr(K)$ 
\item[(ii)] $\int |\rho_K(x)| \,dx \leq \tr|K|$
\end{itemize}
\end{lemma}

\begin{proof}
The map $\rho$ is unique because, by linearity and continuity, it is
completely determined by its action on operators of rank one. For the proof of existence 
let $J\in C_0^\infty(\R^n)$ be a real, non-negative and even function
with $\int J(x) \,dx = 1$. Let $J_\eps(y) := \eps^{-n}J(y/\eps)$ and $J_{\eps,x}(y) := J_\eps(y-x)$. Then
\begin{equation}\label{rho0}
\sprod{J_{\eps,x}}{\ph} = \int J_\eps(x-y) \ph(y) \,dy = (J_\eps * \ph)(x).
\end{equation}
For given $K\in \TT_1(L^2(\R^n))$ we define $\rho_K\in L^1(\R^n)$ by 
\begin{align}
    \rho_K &:=L^1-\lim_{\eps\to 0} \rho_{K,\eps} \label{rho1}\\
    \rho_{K,\eps}(x) &:= \sprod{J_{\eps,x}}{KJ_{\eps,x}}.\nonumber
\end{align}
We claim that the limit \eqref{rho1} exists and that 
\begin{align}
   \rho_K(x) &= \sum_{n\geq 0} \lambda_n \ph_n(x)\overline{\psi_n(x)}\nonumber\\
   \text{if}\quad K &= \sum_{n\geq 0} \lambda_n |\ph_n\rangle \langle \psi_n|\label{rho2}
\end{align}
denotes the singular value decomposition of $K$, that is $\lambda_n
\geq 0$ and $\ph_n, \psi_n$ are orthonormal bases. Using \eqref{rho0},
\eqref{rho2}, $\|J_{\eps}*\ph\|\leq \|J_{\eps}\|_1\|\ph\| = \|\ph\|$
and $\sum_n\lambda_n=\tr|K|$, we see that
\begin{equation}\label{rho3}
   \|\rho_{K,\eps}\|_1 \leq \tr|K|, \qquad \text{for all}\ \eps>0.
\end{equation}
Hence if the limit \eqref{rho1} exists, then (ii) holds and $K\mapsto
\rho_K$ is a continuous linear map. Moreover, if $K_N$ denotes the
$N$-th partial sum of \eqref{rho2}, then, by \eqref{rho3},
$$
   \|\rho_{(K-K_N),\eps}\|_1 \leq \tr |K-K_N| \to 0,\qquad (N\to \infty),
$$
uniformly in $\eps>0$. Therefore it suffices to prove the existence of $\rho_{K_N}$, which follows if we prove that 
$\rho$ maps the rank one operator $|\ph\rangle\langle\psi|$ to $\ph(x)\overline{\psi(x)}$. Indeed,
\begin{eqnarray*}
\lefteqn{\norm{(J_\eps * \ph)\overline{(J_\eps * \psi)}-\ph\overline{\psi}}_1} \\
& = & \norm{(J_\eps * \ph- \ph)\overline{\psi} + (J_\eps * \ph)\left( \overline{(J_\eps * \psi)}- \overline{\psi} \right)}_1 \\
& \leq & \norm{J_\eps * \ph- \ph} \norm{\psi} + \norm{J_\eps * \ph} \norm{J_\eps * \psi - \psi} \to 0 \quad (\eps \to 0)
\end{eqnarray*}
where we used that $J_\eps * \ph \to \ph$ in $L^2(\R^n)$. It now remains to prove (i). This follows from 
$$
\int \rho_K(x) \,dx = \sum_{n\geq 0} \lambda_n \int \ph_n(x) \overline{\psi_n(x)} \,dx = \sum_{n\geq 0} \lambda_n \sprod{\psi_n}{\ph_n}
$$
$$
\tr K = \sum_{n\geq 0} \sprod{\psi_n}{K \psi_n} = \sum_{n\geq 0} \lambda_n \sprod{\psi_n}{\ph_n}. 
$$

\end{proof}

%\bibliographystyle{plain}
%\bibliography{HF}

\begin{thebibliography}{10}

\bibitem{AFGST}
W.~H. Aschbacher, J.~Fr{\"o}hlich, G.~M. Graf, K.~Schnee, and M.~Troyer.
\newblock Symmetry breaking regime in the nonlinear {H}artree equation.
\newblock {\em J. Math. Phys.}, 43(8):3879--3891, 2002.

\bibitem{Aschbacher2009}
Walter~H. Aschbacher and Marco Squassina.
\newblock On phase segregation in nonlocal two-particle {H}artree systems.
\newblock {\em Cent. Eur. J. Math.}, 7(2):230--248, 2009.

\bibitem{Bach1992}
Volker Bach.
\newblock Error bound for the {H}artree-{F}ock energy of atoms and molecules.
\newblock {\em Comm. Math. Phys.}, 147(3):527--548, 1992.

\bibitem{Cances2000}
E.~Canc{\`e}s.
\newblock {SCF} algorithms for {H}artree-{F}ock electronic calculations.
\newblock In M.~Defranceschi and C.~Le~Bris, editors, {\em Lecture Notes in
  Chemistry}, volume~74, pages 17--43. Springer-Verlag, 2000.

\bibitem{CB2000}
Eric Canc{\`e}s and Claude Le~Bris.
\newblock On the convergence of {SCF} algorithms for the {H}artree-{F}ock
  equations.
\newblock {\em M2AN Math. Model. Numer. Anal.}, 34(4):749--774, 2000.

\bibitem{AquaSolo2010}
Anna Dall'Acqua and Jan~Philip Solovej.
\newblock Excess charge for pseudo-relativistic atoms in {H}artree-{F}ock
  theory.
\newblock {\em Doc. Math.}, 15:285--345, 2010.

\bibitem{MelEns2008a}
M.~Enstedt and M.~Melgaard.
\newblock Existence of a solution to {H}artree-{F}ock equations with decreasing
  magnetic fields.
\newblock {\em Nonlinear Anal.}, 69(7):2125--2141, 2008.

\bibitem{MelEns2009}
M.~Enstedt and M.~Melgaard.
\newblock Existence of infinitely many distinct solutions to the
  quasirelativistic {H}artree-{F}ock equations.
\newblock {\em Int. J. Math. Math. Sci.}, pages Art. ID 651871, 20, 2009.

\bibitem{MelEns2008b}
Mattias Enstedt and Michael Melgaard.
\newblock Non-existence of a minimizer to the magnetic {H}artree-{F}ock
  functional.
\newblock {\em Positivity}, 12(4):653--666, 2008.

\bibitem{HLS}
Christian Hainzl, Mathieu Lewin, and Jan~Philip Solovej.
\newblock The mean-field approximation in quantum electrodynamics: the
  no-photon case.
\newblock {\em Comm. Pure Appl. Math.}, 60(4):546--596, 2007.

\bibitem{Helgaker}
T.~Helgaker, P.~Jorgensen, and J.~Olsen.
\newblock {\em Molecular Electronic-Structure Theory}.
\newblock Wiley, first edition, 2000.

\bibitem{HubSied2007}
Matthias Huber and Heinz Siedentop.
\newblock Solutions of the {D}irac-{F}ock equations and the energy of the
  electron-positron field.
\newblock {\em Arch. Ration. Mech. Anal.}, 184(1):1--22, 2007.

\bibitem{Lenzmann2009}
Enno Lenzmann.
\newblock Uniqueness of ground states for pseudorelativistic {H}artree
  equations.
\newblock {\em Anal. PDE}, 2(1):1--27, 2009.

\bibitem{Lieb1981}
Elliott~H. Lieb.
\newblock Variational principle for many-fermion systems.
\newblock {\em Phys. Rev. Lett.}, 46(7):457--459, 1981.

\bibitem{Lieb1984}
Elliott~H. Lieb.
\newblock Bound on the maximum negative ionization of atoms and molecules.
\newblock {\em Phys. Rev. A}, 29(6):3018--3028, Jun 1984.

\bibitem{LS1977}
Elliott~H. Lieb and Barry Simon.
\newblock The {H}artree-{F}ock theory for {C}oulomb systems.
\newblock {\em Comm. Math. Phys.}, 53(3):185--194, 1977.

\bibitem{Lions1987}
P.-L. Lions.
\newblock Solutions of {H}artree-{F}ock equations for {C}oulomb systems.
\newblock {\em Comm. Math. Phys.}, 109(1):33--97, 1987.

\bibitem{ReedSimon2}
Michael Reed and Barry Simon.
\newblock {\em Methods of modern mathematical physics. {II}. {F}ourier
  analysis, self-adjointness}.
\newblock Academic Press [Harcourt Brace Jovanovich Publishers], New York,
  1975.

\bibitem{Reeken1970}
M.~Reeken.
\newblock General theorem on bifurcation and its application to the {H}artree
  equation of the helium atom.
\newblock {\em J. Mathematical Phys.}, 11:2505--2512, 1970.

\bibitem{RuskStill}
Mary~Beth Ruskai and Frank~H. Stillinger.
\newblock Binding limit in the {H}artree approximation.
\newblock {\em J. Math. Phys.}, 25(6):2099--2103, 1984.

\bibitem{Solovej2003}
Jan~Philip Solovej.
\newblock The ionization conjecture in {H}artree-{F}ock theory.
\newblock {\em Ann. of Math. (2)}, 158(2):509--576, 2003.

\end{thebibliography}

\end{document}